\documentclass{l4dc2024}
\hypersetup{breaklinks=true,
colorlinks=true,
}

\usepackage{hyperref}       
\usepackage{booktabs}       

\usepackage{bm}             
\usepackage{bbm}

\usepackage{amsfonts}
\usepackage{amscd}

\usepackage{mathtools}
\usepackage{mathrsfs}

\usepackage{tikz}
\usepackage{xcolor}

\usepackage{nicefrac}       
\usepackage{microtype}      

\usepackage{cleveref}

\usepackage{times}
\usepackage{xspace}
\usepackage{enumitem}

\usepackage{stmaryrd}

\newcommand{\todo}[1]{\textcolor{green}{\text{[TO DO: } #1\text{]}}} 
\newcommand{\tk}[1]{\textcolor{red}{\text{[TK: }#1\text{]}}}

\renewcommand{\u}{\mathsf{u}} 
\newcommand{\w}{\mathsf{w}} 
\newcommand{\x}{\mathsf{x}} 

\newcommand{\F}{\mathcal{F}} 
\newcommand{\G}{\mathcal{G}} 
\newcommand{\K}{\mathcal{K}} 
\newcommand{\T}{\mathcal{T}} 
\newcommand{\M}{\mathcal{M}} 
\renewcommand{\L}{\mathcal{L}} 
\newcommand{\CC}{\mathcal{C}} 
\newcommand{\I}{\mathcal{I}} 

\newcommand{\ufin}{\mathbf{u}} 
\newcommand{\wfin}{\mathbf{w}} 
\newcommand{\xfin}{\mathbf{x}} 

\newcommand{\Ffin}{\mathbf{F}} 
\newcommand{\Gfin}{\mathbf{G}} 
\newcommand{\Kfin}{\mathbf{K}} 
\newcommand{\Tfin}{\mathbf{T}} 
\newcommand{\CCfin}{\mathbf{C}} 
\newcommand{\Ifin}{\mathbf{I}} 

\newcommand{\II}{\mathbb{I}}

\renewcommand{\H}{\mathfrak{H}} 
\renewcommand{\P}{\mathfrak{P}} 
\newcommand{\W}{\mathfrak{W}} 
\newcommand{\causal}{\H_2^{+}} 

\newcommand{\Was}{\operatorname{W_2}}

\newcommand{\cost}{\operatorname{cost}}
\newcommand{\regret} {\operatorname{Regret}}

\newcommand{\op}{\operatorname{op}} 

\newcommand{\ejw}{\e^{{j}\omega}}

\newtheorem{assumption}[theorem]{Assumption}
\newtheorem{problem}[theorem]{Problem}

\newcommand{\defeq}{\coloneqq}                      

\newcommand{\+}{\!+\!}
\renewcommand{\-}{\!-\!}
\renewcommand{\=}{\!=\!}

\renewcommand{\>}{\!>\!}

\newcommand{\ie}{\textit{i.e.}}
\newcommand{\QED}{\hfill $\blacksquare$}

\newcommand{\Z}{\mathbb{Z}}     
\newcommand{\R}{\mathbb{R}}     
\newcommand{\C}{\mathbb{C}}     

\newcommand{\Prob}{\mathfrak{P}} 

\newcommand{\pr}[1]{\left({#1}\right)}          
\newcommand{\br}[1]{\left[{#1}\right]}          
\newcommand{\cl}[1]{\left\{{#1}\right\}}        

\newcommand{\abs}[1]{\vert{#1}\vert}                    

\newcommand{\norm}[2][\text{}]{\Vert{#2}\Vert_{#1}}     

\newcommand{\clf}[1]{\mathcal{#1}} 

\newcommand{\tr}{\operatorname{tr}}         
\newcommand{\Tr}{\operatorname{Tr}}         

\newcommand{\tp}{\intercal}                 

\newcommand{\inv}{\mathrm{\-1}}             

\newcommand{\psdg}{\succ}          

\newcommand{\spect}{\operatorname{sp}} 

\newcommand{\E}{\operatorname{\mathbb{E}}} 
\renewcommand{\Pr}{\operatorname{\mathbb{P}}} 

\newcommand{\sampled}[1][\text{}]{\stackrel{#1}{\sim}}

\newcommand{\supp}{\operatorname{supp}}             

\newcommand{\beq}[1]{\begin{align*}\label{eq:#1}}
\newcommand{\eeq}{\end{align*}}

\newcommand{\suml}{\sum\nolimits}

\newcommand{\e}{\mathrm{e}}             
\newcommand{\half}{\frac{1}{2}} 

\makeatletter
\newcommand{\xMapsto}[2][]{\ext@arrow 0599{\Mapstofill@}{#1}{#2}}
\def\Mapstofill@{\arrowfill@{\Mapstochar\Relbar}\Relbar\Rightarrow}
\makeatother

\usepackage[ruled]{algorithm2e}
\usepackage{comment}

\title[Wasserstein Distributionally Robust Regret Optimal Control]{Wasserstein Distributionally Robust Regret-Optimal Control in the Infinite-Horizon}

 \coltauthor{\Name{Taylan Kargin}$^{\mathbf{*}}$ \Email{tkargin@caltech.edu}\\
  \Name{Joudi Hajar}$^{\mathbf{*}}$ \Email{jhajar@caltech.edu}\\
  \Name{Vikrant Malik}$^{\mathbf{*}}$ \Email{vmalik@caltech.edu}\\
  \Name{Babak Hassibi} \Email{hassibi@caltech.edu}\\
  \addr $^*$Equal contribution \\ Department of Electrical Engineering, Caltech, Pasadena, CA 91125}

\begin{document}
\maketitle
\vspace{-2mm}
\begin{abstract}
We investigate the Distributionally Robust Regret-Optimal (DR-RO) control of discrete-time linear dynamical systems with quadratic cost over an infinite horizon. 
Regret is the difference in cost obtained by a causal controller and a clairvoyant controller with access to future disturbances.
We focus on the infinite-horizon framework, which results in stability guarantees. In this DR setting, the probability distribution of the disturbances resides within a Wasserstein-2 ambiguity set centered at a specified nominal distribution. 
Our objective is to identify a control policy that minimizes the worst-case expected regret over an infinite horizon, considering all potential disturbance distributions within the ambiguity set. 
In contrast to prior works, which assume time-independent disturbances, we relax this constraint to allow for time-correlated disturbances, thus \emph{actual} distributional robustness.
While we show that the resulting optimal controller is non-rational and lacks a finite-dimensional state-space realization, we demonstrate that it can still be uniquely characterized by a finite dimensional parameter. Exploiting this fact, we introduce an efficient numerical method to compute the controller in the frequency domain using fixed-point iterations. This method circumvents the computational bottleneck associated with the finite-horizon problem, where the semi-definite programming (SDP) solution dimension scales with the time horizon. Numerical experiments demonstrate the effectiveness and performance of our framework. 

\end{abstract}


\section{{Introduction}} \label{sec:intro}
Ensuring reliable and effective operation in the face of uncertainty is a fundamental challenge in decision-making and control. Control systems are inherently subject to diverse uncertainties, including exogenous disturbances, measurement inaccuracies, modeling discrepancies, and temporal variations in the underlying dynamics \citep{grinten1968uncertainty,doyle1985structured}. Disregarding these uncertainties during controller design can lead to significant performance degradation and even unsafe and unintended behavior \citep{samuelson2017}.

Traditionally, stochastic and robust control frameworks have addressed this issue primarily through the lens of exogenous disturbances \citep{kalman_new_1960,zames_feedback_1981,doyle_state-space_1988}. Stochastic control, exemplified by Linear–quadratic–Gaussian (LQG), or $H_2$, control, aims to minimize the expected cost, assuming disturbances are generated randomly from a known probability distribution \citep{blackbook}. However, in practice, the true distribution is often estimated from sampled data, rendering this approach vulnerable to inaccurate models. Robust control, on the other hand, seeks to minimize the worst-case cost across a set of potential disturbance realizations, like those with bounded energy or power ($H_\infty$ control) \citep{zhou_robust_1996}. While this guarantees robustness against adversarial disturbances, it can be overly conservative, discarding potentially valuable statistical information. To address this issue, two recent approaches have emerged.
\enlargethispage{\baselineskip}

\paragraph{Regret-Optimal (RO) Control.}Introduced by \cite{sabag2021regret, goel2023regret}, this framework provides a promising approach to address both stochastic and adversarial uncertainties. It defines regret as the performance gap between a causal control policy and a clairvoyant, non-causal policy with perfect knowledge of the system and future disturbances. In the full-information Linear-Quadratic Regulator (LQR) setting, RO controllers minimize the worst-case regret across all bounded energy disturbances \citep{sabag2021regret, goel2023regret}. Furthermore, the infinite-horizon RO controller admits a state-space form, rendering this approach amenable to efficient real-time computation \citep{sabag2021regret}.

Extensions of this framework have been explored for the measurement-feedback setting \citep{ROMF, hajar_regret-optimal_2023},  the dynamic environment setting \citep{goel_regret-optimal_2021-1}, safety critical control \citep{martin2022safe,didier2022system}, filtering \citep{sabag_regret-optimal_2022, goel2023regret}, and distributed control \citep{martinelli_closing_2023}. While these controllers closely track the performance of the non-causal controller in the worst-case disturbance setting, they can, however, become overly conservative when dealing with stochastic disturbances.
\paragraph{Distributionally Robust (DR) Control.}This framework, on the other hand, addresses uncertainty in system dynamics and disturbances by considering ambiguity sets, \ie, a set of plausible probability distributions, rather than considering a single distribution as in $H_2$ or worst-case realization of disturbances as in $H_\infty$ and RO control \citep{yang2020wasserstein,kim_distributional_2021, hakobyan2022wasserstein, tacskesen2023distributionally, aolaritei_wasserstein_2023,aolaritei_capture_2023}.
This approach seeks to design controllers that perform well across all probability distributions of disturbances within an ambiguity set.
The size of the ambiguity set allows one to control the amount of desired robustness against distributional uncertainty so that, unlike $H_\infty$ and RO controllers, the resulting controller is not overly conservative.

While various distributional mismatch measures like total variation \citep{tzortzis_dynamic_2014,tzortzis_robust_2016} and KL divergence \citep{liu_data-driven_2023} are considered in DR control, ambiguity sets are commonly chosen as Wasserstein-2 balls around a nominal distribution due to computational tractability \citep{mohajerin_esfahani_data-driven_2018,gao_distributionally_2022}. Therefore, this approach provides a tractable means to bridge the gap between the realms of stochastic and adversarial uncertainties. 

\subsection{Contributions}
In this work, we consider the Wasserstein-2 distributionally robust regret-optimal (DR-RO) control framework introduced by \cite{DRORO} for the full-information LQR setting and extended by \cite{hajar_wasserstein_2023} to the partial-observability one. DR-RO control aims to design controllers that minimize the worst-case expected regret across all distributions chosen adversarially within a Wasserstein-2 ambiguity set. We summarize our contributions as follows.
\paragraph{Stabilizing Infinite-Horizon Controller.} Rather than the finite-horizon setting prevalent in the DR control literature \citep{hakobyan2022wasserstein, tacskesen2023distributionally,aolaritei_capture_2023,DRORO,hajar_wasserstein_2023}, we focus on the infinite-horizon DR-RO control in the full-information LQR setting. 
Thus, we are able to provide long-term stability and robustness guarantees.

\paragraph{Robustness to Arbitrarily Correlated Disturbances.} Unlike several prior works which assume time-independence of the disturbances \citep{yang2020wasserstein,kim_distributional_2021, hakobyan2022wasserstein, tacskesen2023distributionally, zhong2023nonlinear,aolaritei_wasserstein_2023,aolaritei_capture_2023}, we do not impose such assumptions so that the resulting controllers are robust against time-correlated and non-Gaussian disturbances, thus better capturing distributional robustness.

\paragraph{Computationally Efficient Controller Synthesis.}
Leveraging a strong duality result, we obtain the exact Karush-Khun-Tucker (KKT) conditions for the worst-case distribution and the optimal causal controller.
We show that, although the resulting controller is non-rational, \ie, it does not admit a finite state-space form, it does admit a non-linear finite-dimensional parametric form.
We exploit this parametric structure and provide a computationally efficient numerical method to compute the optimal DR-RO controller in the frequency-domain via fixed-point iterations.
Prior works focus on finite horizon problems (see \citep{DRORO,hajar_wasserstein_2023, tacskesen2023distributionally}) and are hampered by the fact that they require solving a semi-definite program (SDP) whose size scales with the time horizon. This prohibits their applicability when the time horizon is large. Our approach enables efficient implementation of the infinite-horizon DR-RO controller. 

Note that a recent work \cite{brouillon2023distributionally} studies the constrained infinite-horizon DR control problem under partial information by considering time-correlated disturbances, providing stability guarantees and reducing the problem to a finite convex program. However, unlike our method, this work assumes order $T$ stationarity for the stochastic process, formulates a stationarity control problem, and employs ambiguity sets centered at nominal empirical distributions (a data-driven approach, also seen in \citep{yang2020wasserstein, kim_distributional_2021}).

\section{{Preliminaries and Problem Setup}} \label{sec:prelim}
\subsection{Notations}

From hereon, calligraphic letters ($\K$, $\M$, $\L$, etc.) denote operators. $\I$ is the identity operator. 
Sans serif type letters ($\x$, $\u$, $\w$, etc.) denote infinite sequences. Boldface letters ($\Kfin$, $\CCfin$, $\mathbf{w}$, etc.) denote matrices with finite-horizon.
Asterisk $\M^\ast$ denotes the adjoint of $\M$ and 
$\psdg$ denotes the positive-definite ordering.
$\P(\cdot)$ denotes the space of probability measures over a domain.
$\causal$ stands for {(Hardy-2)} space of causal block Toeplitz operators. 
$\Tr$ is the normalized trace function over block Toeplitz operators such that {$\Tr(\I) = p$}, and $\tr$ is the trace of matrices. 
$\norm[\op]{\cdot}$ and $\norm[F]{\cdot}$ are the operator ($H_\infty$) and Frobenius 
($H_2$) norms for operators, respectively. $\norm{\cdot}$ is the Euclidean norm of vectors.
$\{\M\}_{+}$ and $\{\M\}_{-}$ denote the causal and strictly anti-causal parts of an operator $\M$. 

\subsection{A Linear Dynamical System}
We consider a discrete-time linear time-invariant (LTI) dynamical system expressed in its state-space representation as follows:
\begin{equation}\label{eq:state_space}
\begin{aligned}
    x_{t+1} &= A x_{t} + B_u u_{t} + B_w w_{t+1},
\end{aligned}
\end{equation}
Here, $x_{t} \in \R^{n}$ denotes the state, $u_t \in \R^{d}$ the control input, 
and $w_{t} \in \R^{p}$ the exogenous disturbance process. We assume that $(A,B_u)$ and $(A,B_w)$ are stabilizable.

In the rest of this paper, we adopt an operator-theoretic representation of system dynamics~\eqref{eq:state_space}. To this end, we denote by $\x \defeq \{x_{t}\}_{t\in \Z }$, $\u \defeq \{u_{t}\}_{t\in \Z }$, and $\w \defeq \{w_{t}\}_{t\in \Z }$ the bi-infinite state, control input and disturbance sequences, respectively. For a finite-horizon index set $\II_T \defeq \{\-T, \-T\+1,\dots T\-1, T \}$ with $T>0$, we adopt the notation $\xfin_T \defeq \{x_{t}\}_{t\in \II_T }$, $\ufin_T \defeq \{u_{t}\}_{t\in \II_T }$, and $\wfin_T \defeq \{w_{t}\}_{t\in \II_T }$ to denote the finite-horizon counterparts. Using these definitions, we can represent the infinite-horizon system dynamics~\eqref{eq:state_space} equivalently in operator notation as
\vspace{-0mm}
\begin{equation}\label{eq:operator_form}
    \vspace{-0mm}
    \x = \F \u + \G \w, 
\end{equation}
where $\F$ and $\G$ are bi-infinite strictly causal (\ie, strictly lower triangular) and causal (\ie, lower triangular) time-invariant block Toeplitz operators, respectively, corresponding to the dynamics~\eqref{eq:state_space}. We use $\Ffin_T$ and $\Gfin_T$ to denote the finite-horizon counterparts of $\F$ and $\G$ for the interval $\II_T$.
\paragraph{Controller.} In this paper, we consider linear time-invariant (LTI) disturbance feedback control (DFC) policies $\K: \w \to \u $ in the form
\vspace{-0mm}
\begin{equation}\label{eq:feedback_control}
    \vspace{-0mm}
    \u = \K \w.
\end{equation}
Here, $\mathcal{K} \in \causal$ stands for the controller, a causal and time-invariant block Toeplitz operator mapping past disturbance realizations to control inputs. We define the closed-loop transfer operator as
\vspace{-0mm}
\begin{equation} \label{eq:T_K}
\vspace{-0mm}
    \T_{\K} : \w \mapsto \begin{bmatrix} \x \\ \u \end{bmatrix} \defeq \begin{bmatrix} \F \K + \G \\ \K \end{bmatrix} \w,
\end{equation}
which maps the disturbances to the regulated output and the control input of the system~\eqref{eq:state_space} under a fixed control policy $\K$. We similarly adopt the notations $\Kfin_T$ and $\Tfin_{\Kfin,T}$ to respectively denote the finite-horizon controller and closed-loop transfer matrix for the interval $\II_T$.
\paragraph{Cost.} We assume that the cumulative cost incurred by a control policy $\Kfin_T$ within the time interval $\II_T$ for the disturbance realization {$\wfin_T$} is given by:
\vspace{-0mm}
\begin{equation} \label{eq:cost}
\vspace{-0mm}
    \cost_T(\Kfin_T, {\wfin_T}) \defeq \suml_{t\in \II_T} x_t^\tp Q s_t + u_t^\tp R u_t,
\end{equation}
where $Q,R \psdg 0$. By redefining $x_{t} \leftarrow Q^{\half} x_{t}$ and $u_{t} \leftarrow R^{\half} u_{t}$, we can rewrite the cumulative cost~\eqref{eq:cost} in terms of the closed-loop transfer operator~\eqref{eq:T_K} as $\cost_T(\Kfin_T, {\wfin_T}) = {\wfin_T^\ast {\Tfin_{\Kfin,T}^\ast \Tfin_{\Kfin,T}} \wfin_T}$.  

\subsection{The Regret Measure}
In the full-information setting, it is well-known that there exists a unique optimal non-causal policy, $\K_\circ$, defined as
\vspace{-4mm}
\begin{align} \label{eq:noncausal}
    \vspace{-4mm}
    \K_\circ \defeq - (\I + \F^\ast \F)^{\-1}\F^\ast \G,
\end{align}
that minimizes the infinite-horizon cost, $\lim_{T\to \infty} \frac{1}{\abs{\II_T}} \cost_{T}(\Kfin_T,{\wfin_T})$, and a 
a unique optimal non-causal policy, $\Kfin_{\circ,T}$, defined as
\vspace{-4mm}
\begin{align} \label{eq:noncausal_finite}
    \vspace{-4mm}
    \Kfin_{\circ,T} \defeq - (I + \Ffin_T^\ast \Ffin_T)^{\-1}\Ffin_T^\ast \Gfin_T,
\end{align}
that minimizes the finite-horizon cost \eqref{eq:cost}
for all bounded power disturbance realizations \citep{blackbook, sabag2021regret}. Since a non-causal controller is physically unrealizable, we aim to design a causal control policy that performs as best as the optimal non-causal policy $\K_\circ$, which has access to the entire disturbance trajectory at the outset. 
 
To quantify the disparity in accumulated costs between a causal controller and the optimal non-causal controller , $\Kfin_{\circ,T}$, we define the regret as 
\vspace{-0mm}
\begin{equation}
\begin{aligned}\label{eq:regret}
\vspace{-0mm}
    \regret_T(\Kfin_T,{\wfin_T}) &\defeq \cost_T(\Kfin_T,\wfin_T) \- \cost_T(\Kfin_{\circ,T}, {\wfin_T}) \\
    &= \wfin_T^\ast \pr{\Tfin_{\Kfin,T}^\ast \Tfin_{\Kfin,T} \- \Tfin_{\Kfin_\circ,T}^\ast \Tfin_{\Kfin_\circ,T}} \wfin_T.
\end{aligned}
\end{equation}
Put differently, regret measures the excess cost that a causal controller suffers as a result of not foreseeing the realization of future disturbances. 

In the regret-optimal control framework, the objective is to design a causal controller minimizing the worst-case regret among all {bounded energy} disturbances, formulated as follows:
\begin{problem}[Regret-Optimal Control \citep{sabag2021regret}]\label{prob:regret_optimal}
Find a causal control policy, $\K$, that minimizes the time-averaged {worst-case regret} as $T\to\infty$, \ie, 
\begin{align} \label{eq:regret_optimal}
    \inf_{\K \in \causal} { \lim_{T\to \infty}  \frac{1}{\abs{\II_T}}   \sup_{ {\norm{\w_{T}}^2 \leq 1}} \regret_T(\Kfin_T,\wfin_T)}.
\end{align}
\end{problem}
By leveraging the time-invariance of the dynamics~\eqref{eq:state_space} and the controller~\eqref{eq:feedback_control}, problem~\eqref{eq:regret_optimal} can be equivalently reframed as $\inf_{\K \in \causal} \norm[\op]{ \T_{\K}^\ast \T_{\K} - \T_{\K_\circ}^\ast \T_{\K_\circ}},$
which can be solved by reducing it to a Nehari problem \citep{sabag2021regret}. The resulting controller closely mirrors the non-causal controller's performance under worst-case disturbance but may be overly conservative in stochastic disturbance scenarios.

\subsection{Distributionally Robust Regret-Optimal Control}
This paper explores the distributionally robust regret-optimal control approach, aiming to design a causal controller minimizing the worst-case expected regret within an \emph{ambiguity set} of probability distributions of disturbances. The ambiguity set during the time interval $\II_T$ is described as a Wasserstein-2 ball of radius $r\sqrt{\abs{\II_T}}$ centered around a nominal probability distribution $\Pr_\circ \in \Prob(\R^{p\abs{\II_T}})$, \ie, 
\vspace{-0mm}
\begin{equation}\label{eq:ambiguity_set}
\vspace{-0mm}
    \W_T \defeq \cl{\Pr \in \Prob(\R^{p\abs{\II_T}})  \mid  \Was(\Pr,\, \Pr_\circ) \leq r\sqrt{\abs{\II_T}}}.
\end{equation}
Here, the Wasserstein-2 distance is defined as
\vspace{-0mm}
\begin{equation}\label{eq:wasserstein}
\vspace{-0mm}
    \Was(\Pr_1,\Pr_2)^2 \defeq {\inf_{\pi \in \Pi(\Pr_1,\Pr_2)} \int \norm{w_1 - w_2}^2 \,\pi(dw_1,dw_2)} ,
\end{equation}
where the set $\Pi(\Pr_1,\Pr_2)$ consists of all joint distributions with marginals $\Pr_1$ and $\Pr_2$ \citep{villani_optimal_2009, wassOT2}.
 
In \cite{DRORO,hajar_wasserstein_2023}, the worst-case expected regret incurred by a causal controller $\Kfin_T$ during the time interval $\II_T$ is given by 
\vspace{-0mm}
\begin{equation}\label{eq:finite_horizon_regret}
\vspace{-0mm}
    \sup_{\Pr \in \W_T} \E_{\Pr}\br{\regret_T(\Kfin_T, {\wfin_T})}
\end{equation}
where $\E_{\Pr}$ denotes the expectation under the distribution $\Pr$ such that ${\wfin_T} \!\sampled\! \Pr$. Using this formulation in the finite-horizon, we define the worst-case expected regret in infinite-horizon as follows:
\begin{definition}[Worst-Case Expected Regret]\label{def:worst_case_regret}
The time-averaged worst-case expected regret suffered by a causal control policy, $\K \in \causal$, over an infinite horizon is given by  
\vspace{-0mm}
\begin{equation} \label{eq:worst_case_regret}
\vspace{-0mm}
R(\K, {\w})\defeq\lim_{T\to \infty}  \frac{1}{\abs{\II_T}}    \sup_{\Pr \in \W_T} \E_{\Pr}\br{\regret_T(\Kfin_T, {\wfin_T})}.
\end{equation}  
\end{definition}
Using this definition, we formally cast the infinite-horizon DR-RO control problem as follows:
\begin{problem}[Distributionally Robust Regret-Optimal (DR-RO) Control ]\label{prob:DR-RO}
    Find a causal control policy, $\K$, that minimizes the time-averaged worst-case expected regret~\eqref{eq:worst_case_regret} as $T\to\infty$, \ie,
    \begin{align} \label{eq:DR-RO}
    \vspace{-0mm}
    \inf_{\K \in \causal} { \lim_{T\to \infty}  \frac{1}{\abs{\II_T}}    \sup_{\Pr \in \W_T} \E_{\Pr}\br{\regret_T(\Kfin_T, {\wfin_T})}}.
\end{align}
\end{problem}
In Section~\ref{sec:kkt}, we provide an equivalent formulation to Problem~\ref{prob:DR-RO} in terms of the closed-loop transfer operator by appealing to strong duality.

\section{Main Theoretical Results} \label{sec:kkt}

In this section, we present our main theorems. In Theorem~\ref{thm:strong_duality}, we first establish a strong duality reformulation for the infinite-horizon worst-case expected regret in operator form. Exploiting the dual formulation, we reduce solving Problem~\ref{prob:DR-RO} into solving a suboptimal Problem~\ref{prob:suboptimal_DR_RO}. In Theorem~\ref{thm:suboptimal_DR_RO}, we present the suboptimal controller and argue that it is stabilizing. Due to space constraints, we defer the proofs of our theorems to Appendix~\ref{appdx:proofs}. 

\subsection{Reduction to a suboptimal Problem via Strong Duality}
In the finite-horizon DR-RO problem, Theorem 2 in \cite{DRORO} establishes an equivalent formulation for the worst-case expected regret~\eqref{eq:finite_horizon_regret} as a single-parameter optimization problem via strong duality. In Theorem~\ref{thm:strong_duality}, we establish an analogous dual reformulation for the infinite-horizon worst-case expected regret~\eqref{eq:worst_case_regret} as a single-parameter search problem. For ease of notation and clarity of results, we make the following assumption.
\vspace{-3mm}
\begin{assumption}\label{asmp:nominal}
     For any finite-horizon interval $\II_T$, the nominal distribution, $\w_{\circ,T}\!\sampled\!\Pr_\circ$, is absolutely continuous wrt the Lebesgue measure with $\E_{\Pr_\circ}[\w_{\circ,T} \w_{\circ,T}^\ast] = I$.
     \vspace{-3mm}
\end{assumption}
\begin{theorem}[Strong Duality for \eqref{eq:worst_case_regret}] \label{thm:strong_duality}
Let $\CC_{\K} \defeq \T_{\K}^\ast \T_{\K} - \T_{\K_\circ}^\ast \T_{\K_\circ}$ and $\K \in \causal$ be a a causal and time-invariant policy. Under  assumption~\ref{asmp:nominal}, the infinite-horizon worst-case expected regret~\eqref{eq:worst_case_regret} incurred by $\K$ attains a finite value and is equivalent to the following dual problem:
\begin{equation}\label{eq:dual_wosrt_case_regret}
\vspace{-2mm}
    \inf_{\gamma \geq 0}   \gamma (r^2 - \Tr{\I} ) +\gamma  \Tr{ ( \I - \gamma^{\-1}\CC_{\K} )^{\-1}}\quad  \textrm{s.t.} \quad  \gamma \I \psdg \CC_{\K}.
\end{equation}
Furthermore, the worst-case disturbance, $\w_\star$, can be identified from the nominal disturbance, $\w_\circ$, as $\w_\star = ( \I - \gamma_\star^{\-1}\CC_{\K})^{\-1} \w_{\circ}$ where $\gamma_\star$ is the optimal solution to~\eqref{eq:dual_wosrt_case_regret}, which satisfies the following:
\begin{equation}\label{eq:worst_gamma}
\vspace{-2.5mm}
        \Tr{(( \I - \gamma_\star^{\-1}\CC_{\K})^{\-1} - \I)^2} = r^2.
\end{equation}    
\end{theorem}

Using the dual problem~\eqref{eq:dual_wosrt_case_regret}, we can rewrite the DR-RO problem~\eqref{eq:DR-RO} as 
\begin{equation}\label{eq:DR-RO_reformulation}
\vspace{-1.5mm}
    \inf_{\K \in \causal} \inf_{\gamma \geq 0}   \gamma (r^2 - \Tr{\I} ) +\gamma  \Tr{ ( \I - \gamma^{\-1}\CC_{\K} )^{\-1}} \quad  \textrm{s.t.} \quad  \gamma \I \psdg \CC_{\K}.
    \vspace{-1.5mm}
\end{equation}
By exchanging the infima and fixing $\gamma$, we can first find a suboptimal solution $\K_{\gamma}$ to \eqref{eq:DR-RO_reformulation}. Using the suboptimal solutions $\K_\gamma$, we can search for the optimal $\gamma_\star$ by solving equation~\eqref{eq:worst_gamma}. Therefore, we restrict our attention to the suboptimal DR-RO problem stated below.
\begin{problem}[Suboptimal DR-RO Control] \label{prob:suboptimal_DR_RO}
For a fixed $\gamma >{\gamma}_{\textrm{RO}}\defeq \inf_{\K\in\causal} \norm[\op]{\CC_\K}$, find a causal control policy, $\K_{\gamma}$, that minimizes the suboptimal objective function \eqref{eq:dual_wosrt_case_regret}  i.e.,
    \begin{equation}\label{eq:sub_optimal}
       \inf_{\K \in\causal } \Tr{ ( \I - \gamma^{\-1}\CC_{\K} )^{\-1}} \quad \textrm{s.t.} \quad \gamma \I \psdg \CC_{\K}.
\end{equation}
\end{problem}
\vspace{-1.5mm}
\begin{remark}\label{remark:limiting_r}
Note that as $r\to\infty$, the optimal $\gamma_\star$  approaches the lower bound $\norm[\op]{\CC_\K}$, \ie, the worst-case expected regret~\eqref{eq:worst_case_regret} reaches to the worst-case regret as in \cite{sabag2021regret} and the optimal DR-RO controller recovers the optimal RO controller. The optimal regret, ${\gamma}_{\textrm{RO}}$, acts as a global lower bound on $\gamma$. Conversely, as $r\to 0$, $\gamma_\star \to \infty$, leading the worst-case expected regret~\eqref{eq:worst_case_regret} to nominal expected regret and the optimal DR-RO controller recovers the optimal $H_2$ controller. Adjusting $r$ enables the DR-RO controller to interpolate between the RO and $H_2$ controllers.
\end{remark}

\subsection{Solution for the Suboptimal Problem~\ref{prob:suboptimal_DR_RO}}
In its present form, Problem~\ref{prob:suboptimal_DR_RO} is challenging since the controller appears both in an operator inverse, as well as in the constraint $\gamma\I\!\psdg\!\CC_\K$. An alternative formulation via Fenchel duality follows.

\begin{lemma}[Duality for the Suboptimal Problem~\ref{prob:suboptimal_DR_RO}]\label{thm:dual_suboptimal}
Let $\gamma >{\gamma}_{\textrm{RO}}$ be fixed and let assumption~\ref{asmp:nominal} hold. The $\gamma$-optimal DR-RO control Problem~\ref{prob:suboptimal_DR_RO} is equivalent to the following dual problem
\begin{equation}\label{eq:suboptimal_prob}
    \sup_{\M \psdg 0} \inf_{\K \in\causal} - \Tr(\M) + 2\Tr(\sqrt{\M}) + \gamma ^\inv  \Tr(\CC_{\K} \M ).
\end{equation}
\end{lemma}
The concave-convex problem~\eqref{eq:suboptimal_prob} is more manageable since the inner minimization wrt $\K\in\causal$ can be solved via the Wiener-Hopf technique \citep{kailath_linear_2000}.

Introducing the spectral factorization $\Delta^\ast \Delta = \I + \F^\ast \F$ with causal and causally invertible $\Delta$, we present our second main result in Theorem~\ref{thm:suboptimal_DR_RO}, the solution to the suboptimal DR-RO Problem~\ref{prob:suboptimal_DR_RO}. 
\begin{theorem}[suboptimal DR-RO Controller]\label{thm:suboptimal_DR_RO}
The $\gamma$-sub-optimal DR-RO controller $\K_\gamma$ of DR-RO Problem~\ref{prob:suboptimal_DR_RO} coincides with the saddle point $(\K_{\gamma},\M_{\gamma})$ of the dual problem~\eqref{eq:suboptimal_prob}. Furthermore, let $\L_{\gamma}$ denote the causal and causally invertible spectral factor of $\M_\gamma$ such that $\M_{\gamma} = \L_{\gamma} \L_{\gamma}^\ast$. Then 
$(\K_{\gamma},\M_{\gamma})$ uniquely satisfies the following set of equations:
\begin{align}\label{eq:suboptimal_kkt}
      \K_{\gamma} =  \Delta^\inv\cl{\Delta \K_\circ \L_{\gamma}}_{\!+} \L_{\gamma}^\inv, \quad \text{and} \quad \L_{\gamma}^\ast \L_{\gamma} =  \frac{1}{4}\pr{\I \+ \sqrt{\I + 4 \gamma^\inv \{\Delta\K_{\circ}\L_{\gamma}\}_{\!-}^\ast  \{\Delta\K_{\circ}\L_{\gamma}\}_{\!-}  }}^2.
\end{align}
\end{theorem}
The proof of Theorem~\ref{thm:suboptimal_DR_RO}, given in Appendix~\ref{appdx:proofs}, is built upon the KKT conditions for \eqref{eq:suboptimal_prob} and the Wiener-Hopf technique \citep{kailath_linear_2000}.
Note that for $\gamma >{\gamma}_{\textrm{RO}}$, the worst-case expected regret \eqref{eq:worst_case_regret} is finite, which allows us to present Corollary \ref{thm:stabilizable}.

\begin{corollary}\label{thm:stabilizable}
    For any fixed $\gamma >{\gamma}_{\textrm{RO}}$, the suboptimal controller $\K_\gamma$ stabilizes the system dynamics.  
\end{corollary}

\section{An Algorithm for Irrational Controller Synthesis} \label{sec:numeric}
In section \ref{subsec::subK}, we first show that, in the frequency domain, the KKT conditions~\eqref{eq:suboptimal_kkt} are uniquely determined by a finite-dimensional parameter, $\overline{B}_{\gamma}$. This allows us to argue that the sub-optimal controller is \emph{irrational}, and thus does not admit a finite-dimensional state-space realization. In section \ref{subsec:alg}, for any fixed $\gamma$, we propose a fixed-point iteration to find $\overline{B}_{\gamma}$ and thereby to compute the sub-optimal controller, $K_{\gamma}(\e^{j\omega})$. The optimal $\gamma_\star$, and thus the optimal DR-RO controller $K_{\gamma_\star}(\e^{j\omega})$, can be found by using the bisection method on equation~\eqref{eq:worst_gamma}.

\subsection{Finite-Dimensional Parametrization of the Sub-optimal Controller}\label{subsec::subK}

Defining $S_{\gamma,\-}(\e^{j\omega}) \defeq \{\Delta K_\circ L_\gamma\}_{\-}(\e^{j\omega})$, and $N_\gamma(\e^{j\omega}) \defeq {L_\gamma(\e^{j\omega})}^\ast L_\gamma(\e^{j\omega})$, and using the identity $\{\clf{X}\}_{\+}\=\clf{X}\-\{\clf{X}\}_{\-}$ , we restate the KKT equations~\eqref{eq:suboptimal_kkt} in the frequency domain as follows:
\begin{align}
    K_\gamma(\e^{j\omega}) &= K_{\circ}(\ejw)\- \Delta^\inv (\e^{j\omega}) S_{\gamma,\-}(\e^{j\omega}) L_\gamma^\inv (\e^{j\omega}), \label{eq:K_freq}\\
    N_\gamma(\e^{j\omega}) &= \frac{1}{4}\pr{I \+ \sqrt{I \+ 4 \gamma^\inv S_{\gamma,\-}^\ast(\e^{j\omega})  S_{\gamma,\-}(\e^{j\omega}) }}^{\!2} \label{eq:N_freq}
\end{align}

Furthermore, we define the LQR controller $K_{\textrm{lqr}}\!\defeq\!{(R\+B_u^\ast PB_u)}^{\inv}B^\ast_u P A$ and the closed-loop matrix $A_K\defeq A\-B_u K_{\textrm{lqr}}$ where $P \!\succ \!0$ is the unique stabilizing solution to the LQR Riccati equation $P=Q+A^\ast PA-A^\ast PB_u{(R+B_u^\ast PB_u)}^{-1}B_u^\ast PA.$

In Lemma~\ref{lemma:finiteBl}, we show that the strictly anticausal transfer function $S_{\gamma,\-}(\e^{j\omega})$ admits a finite-dimensional state-space representation.
\begin{lemma}\label{lemma:finiteBl}
Let $\overline{A}\defeq A^\ast_K$, $\overline{D}\defeq A^\ast_K P B_w$, and $\overline{C}\defeq -{(R + B_u^\ast P B_u)}^{-\ast/2}B^\ast_u$. We have that
\begin{equation}
    S_{\gamma,\-}(\e^{j\omega}) = \overline{C}(\e^{\-j\omega}I - \overline{A})^\inv \overline{B}_{\gamma}, \quad \text{where} \quad \overline{B}_{\gamma} \defeq \frac{1}{2\pi} \int_{0}^{2\pi} (I-\e^{j\omega} \overline{A})^\inv \overline{D} L_\gamma(\e^{j\omega}) d\omega.
\end{equation}
\end{lemma}
Notice that the rhs of \eqref{eq:N_freq} for $N_\gamma(\ejw)$ involves the square-root of rational term $ S_{\gamma,\-}(\e^{j\omega})^\ast  S_{\gamma,\-}(\e^{j\omega})$. In general, square root does not preserve rationality. We thus get Corollary~\ref{thm:irrational}.
\begin{corollary}\label{thm:irrational}
    For any fixed $\gamma\in(\gamma_{\textrm{RO}},\infty)$, $N_\gamma(\ejw)$ and the suboptimal DR-RO controller, $K_\gamma(\e^{j\omega})$, are irrational. Thus, $K_\gamma(\e^{j\omega})$ does not admit a finite-dimensional state-space form.
\end{corollary}
Even though $K_\gamma(\e^{j\omega})$ does not admit a finite-dimensional state-space form, Lemma~\ref{lemma:finiteBl} suggests a finite-dimensional {\em parametrization} of $N_\gamma(\e^{j\omega})$ through $\overline{B}_{\gamma}$. Theorem~\ref{thm:fixed-point} establishes that $\overline{B}_{\gamma}$ uniquely determines $N_\gamma(\e^{j\omega})$, and thus the suboptimal controller $K_\gamma(\e^{j\omega})$. 
\begin{theorem}[Fixed-Point Solution]\label{thm:fixed-point}
Fix $\gamma \> \gamma_{\textrm{RO}}$ and consider the following set of mappings:
\begin{align}
    &F_1: \overline{B} \mapsto S_{\-}(\ejw) \defeq \overline{C}(\e^{\-j\omega}I - \overline{A})^\inv \overline{B} \label{eq:F1} \\
    &F_{2,\gamma}: S_{\-}(\ejw) \mapsto N(\e^{j\omega}) \defeq \frac{1}{4}\pr{I \+ \sqrt{I \+ 4 \gamma^\inv S_{\-}^\ast(\ejw)  S_{\-}(\ejw) }}^{\!2}\label{eq:F2} \\
    &F_3: N(\ejw) \mapsto  L(\ejw), \quad F_4: L(\ejw) \mapsto  \overline{B}\defeq \frac{1}{2\pi} \int_{0}^{2\pi} (I-\e^{j\omega} \overline{A})^\inv \overline{D} L(\e^{j\omega}) d\omega.
\end{align}
where $F_3$ returns a unique spectral factor of $N(\ejw)>0$. The composition $F_4 \!\circ\! F_3 \!\circ\!  F_{2,\gamma} \!\circ\!  F_1:\overline{B} \mapsto \overline{B}$ admits a unique fixed-point $\overline{B}_{\gamma}$, and  $N_{\gamma}(\e^{j\omega})~\defeq~F_{2,\gamma} \!\circ\!  F_1 (\overline{B}_{\gamma})$  satisfies the KKT conditions~\eqref{eq:suboptimal_kkt}.
\end{theorem}

\subsection{Algorithm Description}\label{subsec:alg}

Motivated by Theorem~\ref{thm:fixed-point}, we introduce Algorithm \ref{alg::1} to compute the suboptimal controller $K_\gamma(\ejw)$ at uniformly sampled points on the unit circle. We start the algorithm with an initial estimate of the parameter $\overline{B}_{\gamma}^{(0)}$ (line 0 in the algorithm). At the $n^\text{th}$ iteration, we construct the functions $S_{\gamma,\-}^{(n)}(\e^{j\omega})$ and $N_{\gamma}^{(n)}(\e^{j\omega})$ from $\overline{B}_{\gamma}^{(n)}$ using the mappings in \eqref{eq:F1}, \eqref{eq:F2}. We subsequently compute the the spectral factor $L_\gamma^{(n)}(\e^{j\omega})$ (line 3) at uniformly sampled points on the unit circle,
from which we compute the next iterate $\overline{B}_{\gamma}^{(n+1)}$ via numerical integration of $F_4$. Upon convergence upto a tolerance, we ascertain the \emph{suboptimal} controller $K_\gamma(\e^{j\omega})$ for a fixed $\gamma>\gamma_{\textrm{RO}}$ at every sampled frequency point using~\eqref{eq:K_freq}.

\paragraph{Spectral Factorization}
Since there is no general closed-form formula for spectral factorization of irrational spectra, we can compute $L_{\gamma}(\ejw)$ only at finitely many frequencies.
Focusing on single-input systems, we employ a discrete Fourier transform (DFT) based factorization method by \cite{rino_factorization_1970}, highlighted in Algorithm 2 in Appendix~\ref{appdx:spect}, to approximate  $L_{\gamma}(\ejw)$ at uniformly sampled points on the unit circle. This method, tailored for \emph{scalar} irrational functions, proves efficient as the associated error term, featuring as a multiplicative phase factor, rapidly diminishes with increasing number of samples.

\RestyleAlgo{ruled}
\SetKwComment{Comment}{/* }{ */}

\begin{algorithm}[ht] 
  \caption{\textit {Distributionally Robust Regret-Optimal Control}}\label{alg::1}
  \SetAlgoNlRelativeSize{0}
  {\tiny [0]} Fix $\gamma>\gamma_{\textrm{RO}}$; initialize: $\overline{B}_{\gamma}^{(0)}$, and select $2^k$ equally spaced values of $\omega\in (0,2\pi)$\;
  
  \For{$n\geq 0$}{
    {\tiny[1]} $S_{\gamma,-}^{(n)}(\mathrm{e}^{j\omega}) = \overline{C}(\mathrm{e}^{j\omega}I - \overline{A})^{-1} \overline{B}_{\gamma}^{(n)}$\;
    
    {\tiny [2]} $N_\gamma^{(n)}(\mathrm{e}^{j\omega}) =  \frac{1}{4}\left(I + \sqrt{I + 4 \gamma^{-1} S_{\gamma,-}^{(n)}(\mathrm{e}^{j\omega})^{*}  S_{\gamma,-}^{(n)}(\mathrm{e}^{j\omega})  }\right)^{2}$\;
    
    {\tiny[3]} $L_\gamma^{(n)}(\mathrm{e}^{j\omega}) = \text{Spectral Factorization}( N_\gamma^{(n)}(\mathrm{e}^{j\omega}))$\;
    
    {\tiny[4]} $\overline{B}_{\gamma}^{(n+1)}  = \frac{1}{2\pi} \int_{-\pi}^{\pi} (I-\mathrm{e}^{j\omega} \overline{A})^{-1} \overline{D} L_\gamma^{(n)}(\mathrm{e}^{j\omega}) d\omega$\;
    
    \If{$\norm{\overline{B}_{\gamma}^{(n+1)}-\overline{B}_{\gamma}^{(n)}}<\textrm{tol}$}{
      
      {\tiny[5]} Compute $K_\gamma(\e^{j\omega}) = K_{\circ}(\ejw)- \Delta^{-1} (\e^{j\omega}) S_{\gamma,-}(\e^{j\omega}) L_\gamma^{(n)} (\e^{j\omega})^{-1}$
      
      \textbf{Break}\;
    }
  }
  \vspace{-6mm}
\end{algorithm}

\section{{Experimental Results}} \label{sec:simul}

In this section, we showcase the applicability of the DR-RO controller, and its performance, compared to $ H_2$, $ H_\infty$, and regret-optimal controllers. We focus our investigation on a set of 4 diverse systems from \cite{aircraft}, and, in particular, use the chemical reactor system, [REA4], as our main benchmark. The system has 8 states and is SISO. We perform all experiments using MATLAB, on a Macbook Air with Apple M1 processor and 8 GB of RAM. We specify the nominal distribution as Gaussian, with zero mean and identity covariance. We investigate various values for the radius $r$, and for each solve the optimization problem using the algorithm outlined in section~\ref{sec:numeric}.

For the system [REA4], a comparative analysis of worst-case expected regret cost as defined in \eqref{def:worst_case_regret} is conducted against the $H_2$, $H_\infty$ \cite{blackbook}, and RO \cite{sabag2021regret} controllers, considering the unique worst-case distribution associated with each controller. The results are depicted in Figures \ref{fig:ER} and \ref{fig:perc}. We redo the analysis considering 3 other systems (described in \cite{aircraft}), and we show the results in Table 1 in Appendix~\ref{appdx:sim}. 
Another performance metric considered is the operator norm of $T_K$ minimized by the $H_\infty$ controller, which is expressed, in the frequency domain as: $\|\T_\K \|_{\op}^2 = \max_{0 \leq \omega \leq 2\pi} \sigma_{\max}( T_K^\ast(\ejw)T_K(\ejw))$. This metric is visualized across all frequencies in Figure \ref{fig:R}.

Figures \ref{fig:ER}, and \ref{fig:perc} emphasize the robust performance of the DR controller in minimizing worst-case expected regret under worst-case disturbance conditions for any given parameter $r$. Notably, the DR controller exhibits a versatile nature, closely mirroring the $H_2$ controller for smaller $r$ while converging towards the behavior of the RO controller for larger $r$. This dual capability underscores its adaptability to different robustness requirements, and aligns with the theoretical insights outlined in Remark~\ref{remark:limiting_r}. Moreover, in Figure \ref{fig:R}, the performance of the DR controller exhibits an interpolation between the $H_2$ and RO controllers \emph{across all frequencies}. 

Finally, we note that Algorithm \ref{alg::1}, coupled with the bisection technique, exhibits notable efficiency; the execution time is $5.8$ seconds for [REA4] system, for $r=0.79$.
This highlights the significance of our approach compared to other DR control methods that rely on an SDP that scales with the time-horizon and number of states (e.g., \cite{DRORO} and \cite{hajar_wasserstein_2023} could only address systems with smaller dimensions and a time-horizon of only 10 steps).

\begin{figure}[!ht]
\centering
\subfigure[]{
    \includegraphics[width=0.45\textwidth]{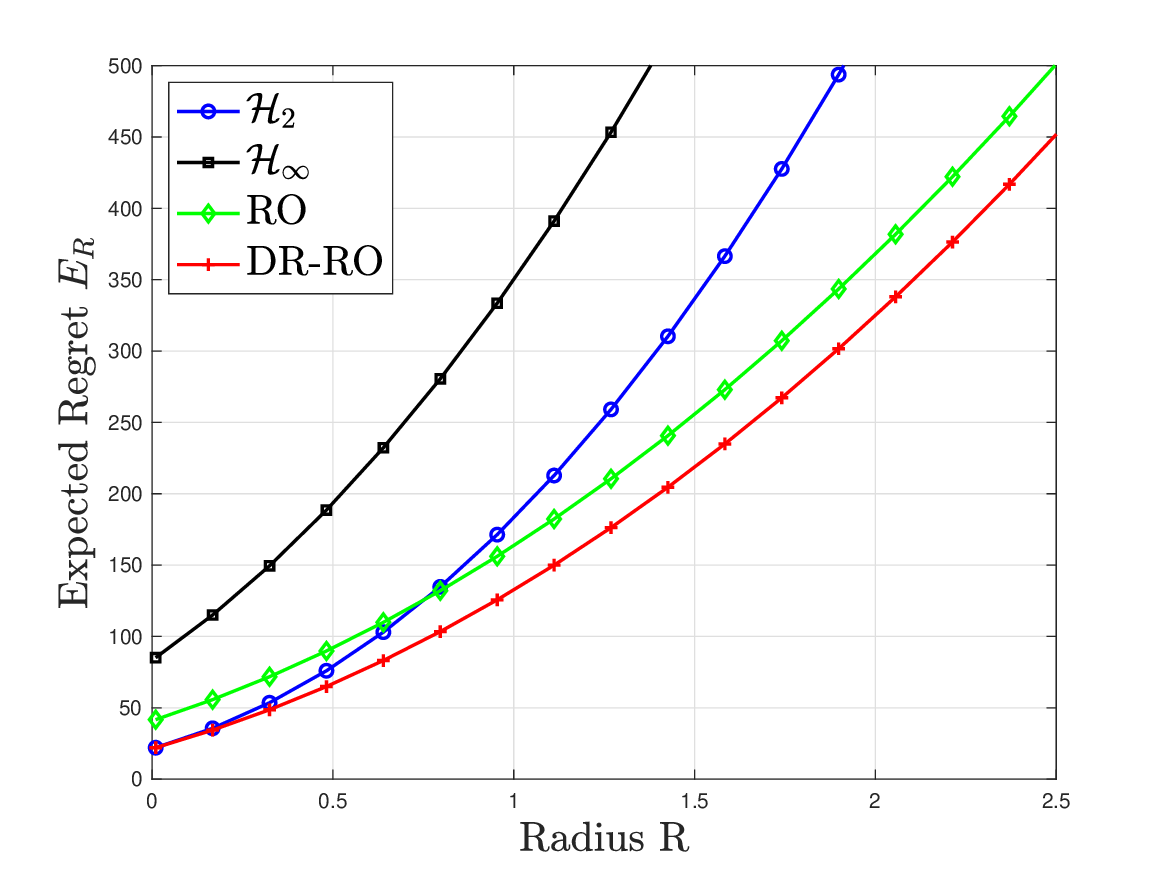}
    \label{fig:ER}
}
\hspace{0.5cm} 
\subfigure[]{
    \includegraphics[width=0.45\textwidth]{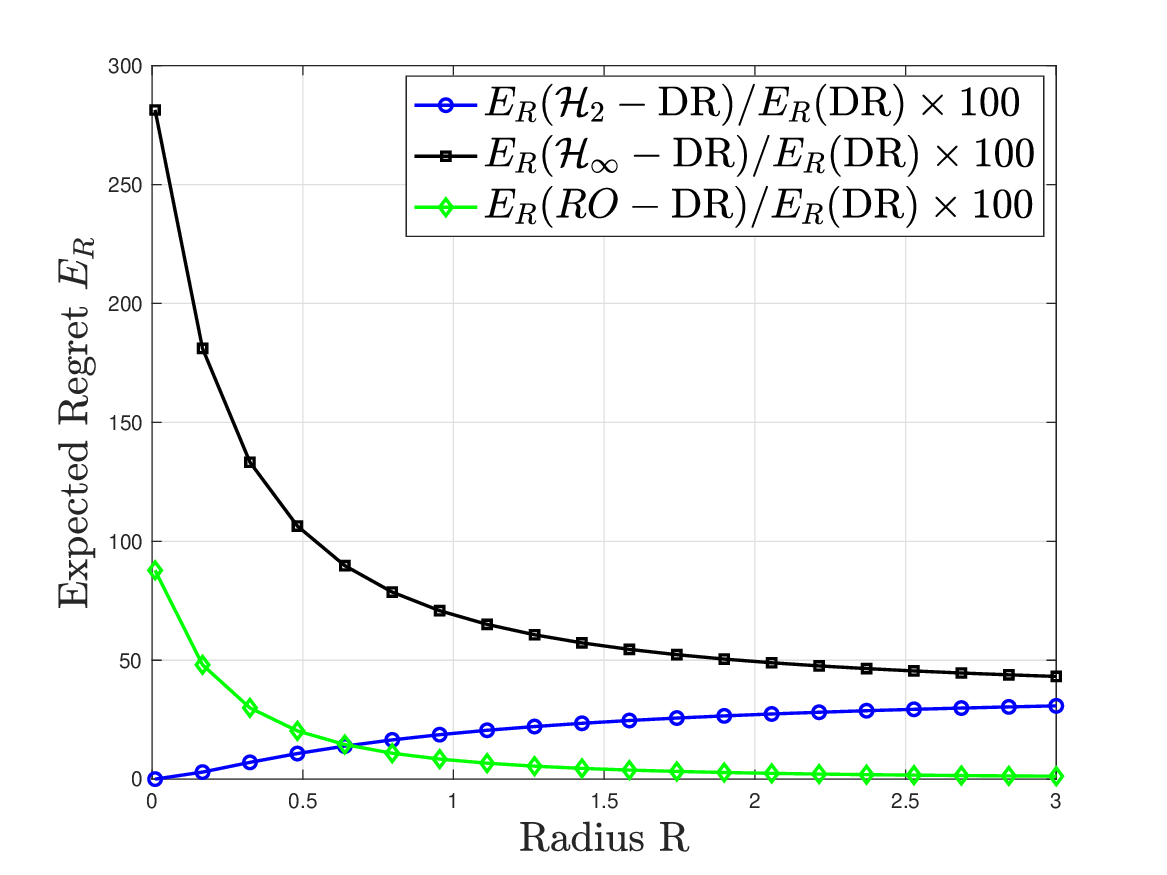}
    \label{fig:perc}
}
\vspace{-3mm}
\caption{(a) The worst-case expected regret cost of each controller for different values of $r$, for system [REA4]. (b) The percentage difference in the worst case regret relative to the DR-RO controller. (a) and (b) show that DR-RO minimizes the cost at all $r's$, and for small (large) $r$, the cost of DR-RO controller matches that of ${\cal H}_2$ (RO). 
The cost of the DR controller is less than that of $H_2$ and RO by $14.5\%$, and of $H_\infty$ by $89.7\%$ for $r=0.639$.}
\end{figure}
\begin{figure}[!ht]
  \vspace{-5mm}
    \centering
    \includegraphics[width=0.6\textwidth]{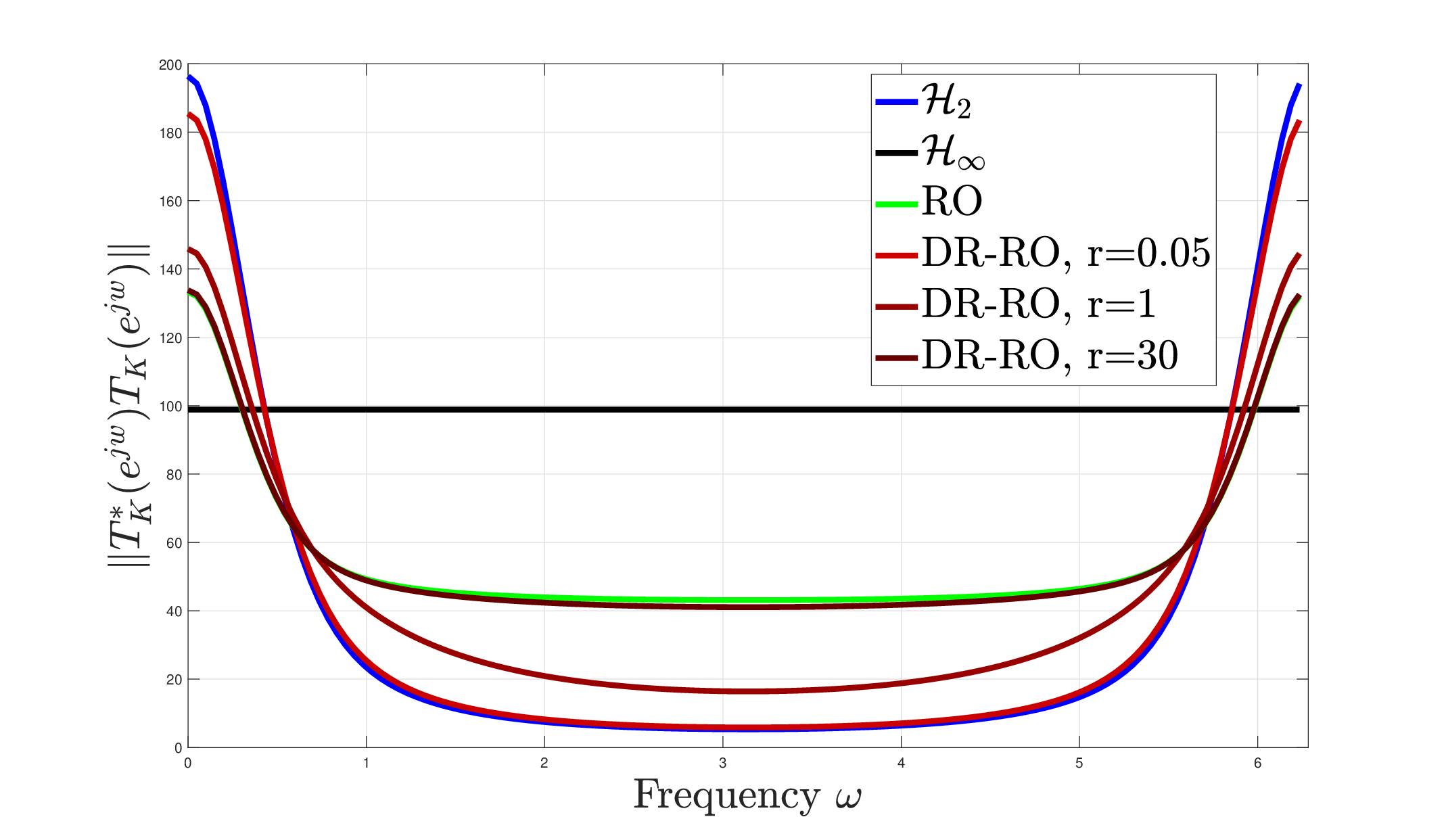}
    \vspace{-3mm}
    \caption{The operator norm, $\|T_K^\ast(\ejw)T_K(\ejw)\|$, of each controller at different frequency values, for system [REA4]. The cost of the DR-RO controller interpolates between $H_2$ and RO according to the value of $r$, across all frequencies. For a small (large) $r$,  DR matches $H_2$ (RO) \emph{across all frequencies}. 
    }
    \label{fig:R}
    \vspace{-5mm}
\end{figure}

\section{{Conclusion and Future Works}} \label{sec:conc}
We studied DR-RO control for discrete-time linear dynamical systems over an infinite horizon. Focusing on regret as a measure of performance introduces a nuanced perspective, while the incorporation of uncertainties within a Wasserstein-2 ambiguity set provides a robust framework for handling unpredictable disturbances. Solving the problem in infinite horizon enhances stability, optimality, and robustness, and aligns the framework with real-world demands, since finite-horizon methods are hampered by SDPs that grow too large.

A key departure from prior research is our deliberate consideration of dependencies among disturbances over time. This approach contrasts with simplistic assumptions of stochastic independence at each time step, and thus better captures the essence of distributional robustness. Even though the optimal controller is irrational, we introduce a computationally efficient numerical method based on fixed-point iterations to find the controller in the frequency domain. Validation through numerical experiments demonstrates the effectiveness of our framework. Looking forward, avenues for future research include finding good low-dimensional rational approximations for the controller, providing convergence guarantees for the fixed point method, extending the algorithm to MIMO systems by exploring irrational matrix spectral factorization \cite{nurdin_new_2005, ephremidze_elementary_2010}, and extending the framework to the partially observable case.


\bibliography{refs} 
\newpage
\appendix
\begin{center}
{\huge Appendix}
\end{center}



\section{Proofs of Theorems} \label{appdx:proofs}
In this section, we give the proofs of the corresponding theorems and lemmas.

\subsection{Proof of Theorem~\ref{thm:strong_duality}}

We prove Theorem~\ref{thm:strong_duality} in multiple steps. First, we give a finite-horizon counterpart of the strong duality result following {\cite{DRORO}}. Second, we rewrite the objective functions of both the finite-horizon and the infinite-horizon dual problems using normalized spectral measure. Third, we show the pointwise convergence of the finite-horizon dual objective function to infinite-horizon objective by analyzing the limiting behavior of spectrum of Toeplitz matrices. Finally, we show that the infinite-horizon dual problem attains a finite value and limits of the optimal value (optimal solution resp.) to finite-horizon dual problem coincide with the optimal value (optimal solution resp.) of the infinite-horizon dual problem.

\paragraph{{Step 1: Finite-Horizon Duality.}} In the first step, we state a strong duality result for finite-horizon problems adapted from {\cite{DRORO}. 
\begin{theorem}[{Strong Duality in the Finite-Horizon [Theorem 2 in \cite{DRORO}]}] \label{thm:finite_duality}Fix $T\>0$, and let $\Kfin_T \in \R^{d\abs{\II_T} \times p\abs{\II_T}}$ be a causal (lower block triangular) finite-horizon controller. Under assumption~\ref{asmp:nominal}, the average worst-case expected regret incurred by $\Kfin_T$ during the time interval $\II_T$ 
\begin{equation}
    \frac{1}{ \abs{\II_T}}\sup_{\Pr \in \W_T} \E_{\Pr}\br{\regret_T(\Kfin_T, {\wfin_T})}
\end{equation}
attains a finite value and is equivalent to the following dual convex problem:
\begin{equation}\label{eq:finite_horizon_dual}
     \inf_{\gamma \geq 0}   \gamma \pr{r^2 - \frac{1}{\abs{\II_T}}\tr{\Ifin_T} } +\gamma  \frac{1}{\abs{\II_T}} \tr\br{ ( \Ifin_T - \gamma^{\-1}\CCfin_{\Kfin,T} )^{\-1}}\quad  \textrm{s.t.} \quad  \gamma \Ifin_T \psdg \CCfin_{\Kfin,T}.
\end{equation}
where $\CCfin_{\Kfin,T} \defeq \Tfin_{\Kfin,T}^\ast \Tfin_{\Kfin,T} - \Tfin_{\Kfin_\circ,T}^\ast \Tfin_{\Kfin_\circ,T}$. Furthermore, the worst-case disturbance, $\wfin_{\star,T}$, can be identified from the nominal disturbance, $\wfin_{\circ,T}$, as $\wfin_{\star,T} = ( \Ifin_T - \gamma_{\star,T}^{\-1}\CCfin_{\Kfin,T})^{\-1}\wfin_{\circ,T}$ where $\gamma_{\star,T}$ is the optimal solution to~\eqref{eq:finite_horizon_dual}, which uniquely satisfies the following equation:
\begin{equation}\label{eq:finite_worst_gamma}
        \frac{1}{\abs{\II_T}}\tr\br{(( \Ifin_T - \gamma_{\star,T}^{\-1}\CCfin_{\Kfin,T})^{\-1} - \Ifin_T)^2} = r^2.
\end{equation}
\end{theorem}

\paragraph{{Step 2: Rewriting Objective Functions via Normalized Spectral Measure.}}
In the next step, we introduce an alternative formulation of the finite-horizon problem in \eqref{eq:finite_horizon_dual} and infinite-horizon problem in \eqref{eq:dual_wosrt_case_regret} in terms of the normalized spectral measures of $\CCfin_{\Kfin,T}$ and $\CC_{\K}$, respectively. To do this, we define the spectrum of the matrix $\CCfin_{\Kfin,T}$,
\begin{align} \label{eq:spectrum}
    \spect(\CCfin_{\Kfin,T}) \defeq \{ \lambda \in \C \mid \lambda \Ifin_T - \CCfin_{\Kfin,T} \, \text{ is singular.}\},
\end{align}
and similarly for the operator $\CC_{\K}$. Note that both $\CCfin_{\Kfin,T}$ and $\CC_{\K}$ are positive definite and their spectrum consist of positive operators since
\begin{align*}
    \CC_{\K} &= (\Delta \K \- \Delta \K_{\circ})^\ast  (\Delta \K \- \Delta \K_{\circ}) \\
    \CCfin_{\Kfin,T} &= (\bm{\Delta}_T \Kfin_{T} \- \bm{\Delta}_T \Kfin_{\circ,T})^\ast (\bm{\Delta}_T \Kfin_{T} \- \bm{\Delta}_T \Kfin_{\circ,T}) 
\end{align*}
where $\bm{\Delta}_T^\ast \bm{\Delta}_T \= \Ifin_T \+ \Ffin_{T}^\ast \Ffin_{T}$ is the block Cholesky factorization with lower block triangular $\bm{\Delta}_T$ and $\bm{\Delta}_T^{\inv}$, and $\Delta^\ast \Delta \= \I \+ \F^\ast \F$ is the block canonical spectral factorization with causal ${\Delta}$ and ${\Delta}^{\inv}$. Furthermore, we denote by
\begin{equation}\label{eq:spect_measure}
    \widehat{\rho}_T \defeq \frac{1}{p \abs{\II_T}}  \suml_{k = 1}^{p\abs{\II_T}} \delta_{\lambda_{k,T}}, 
\end{equation}
the normalized spectral measure of $\CCfin_{\Kfin,T}$ where $\lambda_{k}(\CCfin_{\Kfin,T})$ is the $k^{\text{th}}$ largest eigenvalue of $\CCfin_{\Kfin,T} \in \R^{p\abs{\II_T}\times p\abs{\II_T}}$, for $k=1,\dots, p\abs{\II_T}$, and define the convex function
\begin{align}\label{eq:function} 
    f(\gamma,\lambda) &\defeq \gamma(r^2 - p) + \gamma p(1-\gamma^\inv \lambda )^{\inv}.
\end{align}
Note that the support of the normalized spectral measure of $\CCfin_{\Kfin,T}$ overlaps with its spectrum, \ie, $\supp(\widehat{\rho}_T) = \spect(\CCfin_{\Kfin,T})$. Then, we can rewrite the finite-horizon objective function in \eqref{eq:finite_horizon_dual} as follows
\begin{align}\label{eq:finite_objective}
    F_T(\gamma) \defeq  \begin{cases}
        \int f(\gamma,\lambda) \, \widehat{\rho}_T(d \lambda), &\text{ if } \; \gamma > \sup( \supp(\widehat{\rho}_T) ) \\ 
        +\infty, &\text{ o.w.}
    \end{cases}
\end{align}
Similarly, denoting by $\rho(d \lambda)$ the normalized spectral measure of the operator $\CC_{\K}$, we can rewrite the infinite-horizon objective function in \eqref{eq:dual_wosrt_case_regret} as follows
\begin{align}\label{eq:infinite_objective}
    F(\gamma) \defeq  \begin{cases}
        \int f(\gamma,\lambda) \, {\rho}(d \lambda), &\text{ if } \; \gamma > \sup( \supp({\rho}) ) \\ 
        +\infty, &\text{ o.w.}
    \end{cases}
\end{align}
Given these definitions, our goal is to show that 
\begin{equation}\label{eq:limit_inf}
    \lim_{T\to \infty} \inf_{\gamma \geq 0} F_T(\gamma) = \inf_{\gamma \geq 0} F(\gamma) \; \text{ and } \; \lim_{T\to \infty} \gamma_{\star,T} = \gamma_{\star}.
\end{equation}

\paragraph{{Step 3: Asymptotic Equivalence of the Objective Functions.}}
In this step, we aim to show that $\lim_{T\to \infty} F_T(\gamma) = F(\gamma)$ pointwise. To do so, we first note that the finite-horizon matrix $\CCfin_{\Kfin,T}$ and the finite-dimensional block Toeplitz matrix $\CC_{\K,T}$ of horizon $T$ obtained from the block Toeplitz operator $\CC_{\K}$ are asymptotically equivalent \citep{gray_toeplitz_2005}, \ie,
\begin{equation}\label{eq:asymp_equi}
    \lim_{T\to\infty} \frac{1}{\abs{\II_T}} \tr(\abs{\CCfin_{\Kfin,T}-\CC_{\K,T}}) = 0.
\end{equation}
In other words, the matrix $\CCfin_{\Kfin,T}$ is asymptotically block Toeplitz. Afterwards, we resort to Szego limit theorems for block Toeplitz matrices \citep{gazzah_asymptotic_2001,gutierrez-gutierrez_asymptotically_2008, gutierrez-gutierrez_block_2011}. Denoting by $\rho_T(d \lambda)$ the normalized spectral measure of $\CC_{\K,T}$, there exists a constant $b>0$ such that $\sup(\supp(\rho_T)) \leq b$ for all $T>0$ \citep{gray_toeplitz_2005}. Furthermore, for any continous function $\phi :\R \to \R$, we have that \citep{gazzah_asymptotic_2001}
\begin{equation}\label{eq:weak_conv}
    \lim_{T\to\infty} \int \phi(\lambda) \,\rho_T(d \lambda) =  \int \phi(\lambda)\, \rho(d \lambda).
\end{equation}
Therefore, by the asymptotic equivalance of $\CCfin_{\Kfin,T}$ and $\CC_{\K,T}$ in \eqref{eq:asymp_equi}, we have that 
\begin{equation}\label{eq:weak_conv1}
    \lim_{T\to\infty} \int \phi(\lambda) \,\widehat{\rho}_T(d \lambda) =  \int \phi(\lambda)\, \rho(d \lambda).
\end{equation}
Since $f(\gamma,\lambda)$ in \eqref{eq:function} is continous in both terms, for any $\gamma\geq 0$ such that $F(\gamma)<\infty$, we have that
\begin{equation}\label{eq:weak_conv2}
    \lim_{T\to\infty} \int f(\gamma,\lambda) \,\widehat{\rho}_T(d \lambda) =  \int f(\gamma,\lambda) \,\rho(d \lambda).
\end{equation}

\paragraph{{Step 4: Asymptotic Equivalence of the Infimum and the Solution.}}
In the final step, we show that \eqref{eq:limit_inf} holds. To do so, we use the notion of $\Gamma-$convergence \citep{dal_maso_introduction_1993, braides_gamma-convergence_2002}. We introdcuce the following theorem that will be the key in showing \eqref{eq:limit_inf}.

\begin{theorem}[{$\Gamma-$convergence adapted from \citep{dal_maso_introduction_1993}}]\label{thm:gamma_conv}
Let $F_T: \mathbb{X} \to \R\cup\{\infty\}$ for each $T>0$, and $F: \mathbb{X} \to \R\cup\{\infty\}$ be functions such that 
\begin{itemize}
    \item[i.] for every $x\in\mathbb{X}$, and every $\{x_T\}$ such that $x_T \to x$, it holds that
    \begin{equation}\label{eq:liminf}
        F(x) \leq \liminf_{T \to \infty} F_T(x_T),
    \end{equation}
    \item[ii.] for every $x\in\mathbb{X}$, there exists $\{x_T\}$ such that $x_T \to x$, and
    \begin{equation}\label{eq:limsup}
        F(x) \geq \limsup_{T \to \infty} F_T(x_T).
    \end{equation}
\end{itemize}
Then, we have that
\begin{itemize}
    \item[i.] $\lim_{T\to \infty} \inf_{x\in \mathbb{X}} F_T(x) = \inf_{x\in \mathbb{X}} F(x)$,
    \item[ii.] if $x_{\star,T}\in\mathbb{X}$ is the minimizer of $F_T$ for every $T>0$ and $x_{\star}\in\mathbb{X}$ is the minimizer of $F$, then $\lim_{T\to \infty} x_{\star,T} = x_{\star}$.
\end{itemize}
\end{theorem}

Consider $F_T$ defined in \eqref{eq:finite_objective} and $F$ defined in \eqref{eq:infinite_objective}. To show the condition in \eqref{eq:limsup}, let $\gamma_0 \geq 0$ be such that $F_T(\gamma_0)<\infty$ for all $T>0$ and $F(\gamma_0)<\infty$. For any $\gamma\geq 0$ such that $F(\gamma) < \infty$, let $\gamma_T \defeq (1/T) \gamma_0 + (1-1/T) \gamma $. Then, we have that
\begin{align}
\limsup_{T \to \infty} F_T(\gamma_T) &= \limsup_{T \to \infty} \int f(\gamma_T,\lambda) \, \widehat{\rho}_T(d \lambda), \\
&\leq \limsup_{T \to \infty} \pr{ \frac{1}{T} \int f(\gamma_0,\lambda) \, \widehat{\rho}_T(d \lambda) + \pr{1-\frac{1}{T}} \int f(\gamma,\lambda) \, \widehat{\rho}_T(d \lambda) },  \\
&= \int f(\gamma,\lambda) \, {\rho}(d \lambda) = F(\gamma),
\end{align}
where the second inequality is due to convexity of $f(\gamma,\lambda)$ and the last equality is due to pointwise convergence in \eqref{eq:weak_conv2}.

To show the condition in \eqref{eq:liminf}, let $\gamma\geq 0$ such that $F(\gamma) < \infty$ and $\{\gamma_T\}$ be a sequence such that $\gamma_T \to \gamma$. Defining $\bar{f}_T(\lambda) \defeq f(\gamma_T,\lambda)$ for all $T>0$, and $\bar{f}(\lambda) \defeq f(\gamma,\lambda)$, we have that $\lim_{T\to \infty} \bar{f}_T(\lambda) = \bar{f}(\lambda)$ for any $\lambda$ due to continuity of $f(\gamma,\lambda)$. We use a variant of Fatou's lemma with varying measures, as stated in the following lemma.
\begin{lemma}[{Fatou's Lemma with Varying Measure adapted from \citep{feinberg_fatous_2019}}]\label{thm:fatou}
Let $\mu_T$ for each $T>0$ and $\mu$ be finite measures over $\mathbb{X}$ such that $\supp(\mu_T) \subseteq K$ for all $T>0$ and $\supp(\mu) \subseteq K$ for a compact $K\subseteq \mathbb{X}$ and 
\begin{equation}\label{eq:weak_conv_3}
\lim_{T\to\infty} \int \phi(x) \mu_T(d x) = \int \phi(x) \mu(d x),
\end{equation}
for any continous function $\phi$ defined on $K$. Furthermore, let $f_T$ for each $T>0$ and $f$ be continous functions defined on $K$ such that $\lim_{T\to \infty} f_T(x) =f(x)$ for each $x\in K$. Then, we have that
\begin{equation}\label{eq:fatou}
\int f(x) \mu(d x) \leq \liminf_{T\to\infty} \int f_T(x) \mu_T(d x).
\end{equation}
\end{lemma}
Note that the support of $\widehat{\rho}_T$ are all uniformly bounded and $\widehat{\rho}_T$ converges weakly to $\rho$ as in \eqref{eq:weak_conv2}. Therefore, by Lemma~\ref{thm:fatou}, we have that
\begin{align}
\liminf_{T \to \infty} F_T(\gamma_T) &= \liminf_{T \to \infty} \int \bar{f}_T(\lambda) \, \widehat{\rho}_T(d \lambda), \\
&\geq \int \bar{f}_T(\lambda) \, {\rho}(d \lambda) = F(\gamma).
\end{align}
Therefore, by Theorem~\ref{thm:gamma_conv}, the limits in~\eqref{eq:limit_inf} hold.
\QED

\subsection{Proof of Lemma \ref{thm:dual_suboptimal}}

The convex mapping $\clf{X}\! \mapsto\! \Tr \clf{X}^{\inv}$ for $\clf{X}\!\psdg\! 0$ can be expressed via Fenchel duality as
\begin{equation}\label{eq:trace_of_inverse}
   \sup_{\M \psdg 0 } -\Tr(\clf{X} \M ) + 2\Tr(\sqrt{\M}) = \begin{cases} \Tr(\clf{X}^{-1}), & \quad\text{if } \clf{X}\psdg 0 \\  +\infty, &\quad \text{o.w.} \end{cases}
\end{equation}
Using the identity~\eqref{eq:trace_of_inverse}, we rewrite the original problem \eqref{eq:sub_optimal} as, $$\inf_{\K \in\causal}  \sup_{\M \psdg 0 } -\Tr((\I - \gamma ^\inv \CC_{\K} ) \M ) + 2\Tr(\sqrt{\M}).$$ The proof follows immediately from minimax theorem as the objective function is concave-convex. \QED

\subsection{Proof of Theorem \ref{thm:suboptimal_DR_RO}}

The following lemma will be useful in the proof of Theorem \ref{thm:suboptimal_DR_RO}.
    \begin{lemma}[{Wiener-Hopf Technique \citep{kailath_linear_2000}}]
    Consider the problem of approximating a non causal controller $\K_\circ$ by a causal controller $\K$, such that $\K$ minimises the cost $\Tr(\CC_{\K} \M )$, i.e.,
    \begin{align}
        \inf_{\K \in\causal} \Tr(\CC_{\K} \M )
    \end{align}
    where $\M \psdg 0$, $\CC_{\K} = \left(\K - \K_\circ \right)^{\ast}\Delta^{\ast}\Delta\left(\K - \K_\circ \right)$ and $\K_\circ$ is the non-causal controller that makes the objective defined above zero. Then, the solution to this problem is given by 
    \begin{align}
        \K = \Delta^\inv\cl{\Delta \K_\circ \L}_{\!+} \L^\inv,
        \label{eq::nehari_controller}
    \end{align}
    where $\L$ is the unique causal and causally invertible spectral factor of $\M$ such that $\M = \L \L^\ast$ and $\cl{ \cdot }_{\!+}$ denotes the causal part of an operator.
    \label{lemma::nehari}
\end{lemma}
\begin{proof}
    Let $\L$ be the unique causal and causally invertible spectral factor of $\M$, $\ie$ $\M = \L \L^{\ast}$. Then, using the cyclic property of $\Tr$, the objective can be written as,
    \begin{align}
        \inf_{\K \in\causal} \Tr( \Delta\left(\K - \K_\circ \right) \M \left(\K - \K_\circ \right)^{\ast}\Delta^{\ast}) &= \inf_{\K \in\causal} \Tr( \left(\Delta\K - \Delta\K_\circ \right) \L \L^{\ast}\left(\Delta\K - \Delta\K_\circ \right)^{\ast}) \\
        &= \inf_{\K \in\causal} \Tr( \left(\Delta\K\L - \Delta\K_\circ\L \right) \left(\Delta\K\L - \Delta\K_\circ\L \right)^{\ast}) \\
        &= \inf_{\K \in\causal} \left\| \Delta\K\L - \Delta\K_\circ \L \right\|^2_{F},
    \end{align}
 where $\|A\|^2_F$ represents the square of the frobenium norm of a matrix. 
 Since $\Delta, \K$ and $\L$ are causal, and $\Delta\K_\circ\L$ can be broken into causal and non-causal parts, it is evident that the controller that minimises the objective is the one that makes the term $\Delta\K\L - \Delta\K_\circ \L$ strictly anti-causal, cancelling off the causal part of $\Delta\K_\circ \L$. This means that the optimal controller satisfies,
 \begin{align}
    \Delta\K\L = \cl{\Delta \K_\circ \L}_{\!+}.
\end{align}
Also, since $\L^\inv$ and $\Delta^\inv$ are causal, the optimal causal controller is given by \eqref{eq::nehari_controller}.
\end{proof}

\paragraph{\textit{Proof of Theorem \ref{thm:suboptimal_DR_RO}.}}
We first simplify our optimisation problem \eqref{eq:sub_optimal} using Lemma \ref{thm:dual_suboptimal}. We then find the conditions on the optimal optimisation variables using Karush-Kuhn-Tucker (KKT) conditions.

\paragraph{{Reformulation of the original problem.}}
  Using Lemma \ref{thm:dual_suboptimal}, the original optimisation problem \eqref{eq:sub_optimal} can be written as,
\begin{align*}
       \inf_{\substack{\K \in\causal, \\ \gamma \I \psdg \CC_{\K} }} \Tr(  ( \I - \gamma^\inv \CC_{\K} )^{\-1}) &=  \inf_{\K \in\causal}  \sup_{\M \psdg 0 } -\Tr((\I - \gamma ^\inv \CC_{\K} ) \M ) + 2\Tr(\sqrt{\M}) \\
       &=  \sup_{\M \psdg 0 }   -\Tr(\M)  + 2\Tr(\sqrt{\M}) +\inf_{\K \in\causal}\gamma^\inv\Tr(\CC_{\K} \M )
\end{align*}

\paragraph{{KKT Conditions.}}
Note that this problem is convex in $\mathcal{M}$. Hence, the optimal solution $\mathcal{M}_{\gamma}$ satisfies the following KKT conditions, \cite{blackbook}
\begin{align}
    \frac{d}{d\mathcal{M}} \left[ -\Tr(\M)  + 2\Tr(\sqrt{\M}) +\inf_{\K \in\causal}\gamma^\inv\Tr(\CC_{\K} \M ) \right] = 0.
\end{align}
To calculate this derivative, we will use the following identity \cite{bertsekas2016nonlinear}. If we have,
\begin{align}
    g(x) = \min_y f(x, y) = f(x, \hat{y}(x)),
\end{align}
where $f(x, y)$ is convex and $\hat{y}(x)$ is the optimum solution for each $x$, then,
\begin{align}
    \nabla_x g(x) =\nabla_x f(x, y) \bigg\rvert_{(x, \hat{y}(x))}.
    \label{eq::derivative_nabla}
\end{align}
This means that the optimal solution $\M_{\gamma}$ satisfies the KKT,
\begin{align}
    -\I + \M_{\gamma}^{-\half} + \gamma^\inv \CC_{\K_\gamma} = 0.
    \label{eq::KKT_optimal_M}
\end{align}
Substituting the optimal controller $\K_{\gamma}$ from Lemma \ref{lemma::nehari}, we get,
\begin{align}
    -\I + (\L_{\gamma} \L_{\gamma}^\ast)^{-\half} + \gamma^\inv  \L_{\gamma}^{\!-\ast} \{\Delta\K_{\circ}\L_{\gamma}\}_{\!-}^\ast  \{\Delta\K_{\circ}\L_{\gamma}\}_{\!-} \L_{\gamma}^\inv = 0.
\end{align}
Multiplying by $\L^\ast$ from the left and $\L$ from the right, we get,
\begin{align}
      &-\L_{\gamma}^\ast \L_{\gamma} + \L_{\gamma}^\ast(\L_{\gamma} \L_{\gamma}^\ast)^{-\half} \L_{\gamma} + \gamma^\inv \L_{\gamma}^\ast \L_{\gamma}^{\!-\ast} \{\Delta\K_{\circ}\L_{\gamma}\}_{\!-}^\ast  \{\Delta\K_{\circ}\L_{\gamma}\}_{\!-} \L_{\gamma}^\inv \L_{\gamma} = 0.
\end{align}
But, $(\L_{\gamma}^\ast(\L_{\gamma} \L_{\gamma}^\ast)^{-\half} \L_{\gamma})^2 = \L_{\gamma}^\ast \L_{\gamma}$. Hence, the KKT can be written as,
\begin{align}
      -\L_{\gamma}^\ast \L_{\gamma} + (\L_{\gamma}^\ast \L_{\gamma} )^{\half} + \gamma^\inv \{\Delta\K_{\circ}\L_{\gamma}\}_{\!-}^\ast  \{\Delta\K_{\circ}\L_{\gamma}\}_{\!-} = 0.
\end{align}
On completing the squares, we get,
\begin{align}
      \left( (\L_{\gamma}^\ast \L_{\gamma} )^{\half} - \frac{\I}{2} \right)^2 = \frac{\I}{4} + \gamma^\inv \{\Delta\K_{\circ}\L_{\gamma}\}_{\!-}^\ast  \{\Delta\K_{\circ}\L_{\gamma}\}_{\!-}.
      \label{eq::kkt_square_form}
\end{align}
\paragraph{{Square Root of a Positive Definite Matrix.}}
Note that the matix $\frac{\I}{4} + \gamma^\inv \{\Delta\K_{\circ}\L_{\gamma}\}_{\!-}^\ast  \{\Delta\K_{\circ}\L_{\gamma}\}_{\!-}$ is positive definite. Consider the singular value decomposition of the following matrices,
\begin{align}
       (\L_{\gamma}^\ast \L_{\gamma} )^{\half} - \frac{\I}{2}  &= U \Lambda  U^*, \label{eq::PSD1} \\
       \frac{\I}{4} + \gamma^\inv \{\Delta\K_{\circ}\L_{\gamma}\}_{\!-}^\ast  \{\Delta\K_{\circ}\L_{\gamma}\}_{\!-} &= V \Gamma V^*.
\end{align}

Now, \eqref{eq::kkt_square_form} gives, $\Lambda^2 = \Gamma$.
Note that the diagonal entries of $\Gamma$ are greater than $\frac{1}{4}$. Thus, the absolute value of the diagonal entries of $\Lambda$ are greater than $\frac{1}{2}$. But, since $(\L_{\gamma}^\ast \L_{\gamma} )^{\half}$ is positive definite, we have, $\Lambda = \sqrt{\Gamma}$ due to \eqref{eq::PSD1}. Hence, we can take the unique positive definite square root of the right hand side of \eqref{eq::kkt_square_form} to get,
\begin{align}
       (\L_{\gamma}^\ast \L_{\gamma} )^{\half} - \frac{\I}{2} = \sqrt{\frac{\I}{4} + \gamma^\inv \{\Delta\K_{\circ}\L_{\gamma}\}_{\!-}^\ast  \{\Delta\K_{\circ}\L_{\gamma}\}_{\!-}}.
\end{align}
\QED

\subsection{Proof of Lemma~\ref{lemma:finiteBl}}
Denote by $L_{\gamma}(\e^{j\omega})$ the frequency-domain counterpart of the causal operator $\L_\gamma$ where 
\begin{align}\label{eq:L}
L_\gamma(\e^{j\omega})=\sum_{t=0}^{\infty} l_t e^{-j\omega t}.
\end{align}
From lemma 4 in \cite{sabag2021regret}, we have that $\Delta(\ejw) K_\circ(\ejw)$ can be written as the sum of a causal and strictly anticausal transfer functions:
\begin{align}
     &\Delta(\ejw) K_\circ(\ejw)= T(\ejw)+U(\ejw)\label{eq:deltaK}\\
 &T(\ejw)= \bar{C}{(e^{-j\omega}I-\bar{A})}^{-1}\bar{D}, \quad U(\ejw)=\bar{C}P (A {(\ejw I-A)}^{-1} + I)B_w\label{eq:TU}
\end{align}
where $\bar{A}$, $\bar{D}$, and $\bar{C}$ are as defined in the statement of Lemma \ref{lemma:finiteBl}, and $P$ the unique stabilizing solution to the Riccati equation previously mentioned in \ref{subsec::subK}.
Then, using equations \eqref{eq:deltaK} and \eqref{eq:TU}, $S_{\gamma,\-}$ , previously defined as $S_{\gamma,\-}(\e^{j\omega}) = \cl{\Delta K_\circ L_\gamma}_{\-}(\e^{j\omega})$, can be written as: 
\begin{align}
     S_{\gamma,\-}(\ejw) &=  \cl{T  L_\gamma}_{\-}(\ejw) +  \cl{U L_\gamma}_{\-}(\ejw)\\
    &\stackrel{(a)}{=} \cl{\bar{C}(\ejw I - \bar{A})^\inv \bar{D}L_\gamma(\ejw)}_{\-}\\
    &\stackrel{(b)}{=} \cl{\bar{C}\sum_{t=0}^{\infty} e^{j\omega (t+1)} \bar{A}^t \bar{D} \sum_{m=0}^{\infty}l_m e^{-j\omega m}}_{\-}\label{eq:ac}\\
    &\stackrel{(c)}{=} {\bar{C}\left(\sum_{t=0}^{\infty} e^{j\omega(t+1)} \bar{A}^t \right) \left(\sum_{m=0}^{\infty}\bar{A}^m \bar{D}l_m\right) }.
\end{align}
Here, (a) holds because both $U(\ejw)$ and $L_\gamma(\ejw)$ are causal, so the strictly anticausal part of $U(\ejw) L_\gamma(\ejw)$ is zero. (b) holds as we do the Neumann series expansion of $(\ejw I-\bar{A})$ and replace $L_\gamma(\ejw)$ by its equation \eqref{eq:L}. (c) holds as we take the anticausal part of expression \eqref{eq:ac} to be the strictly positive exponents of $\ejw$. We complete the proof by using again the Neuman series, defining $\bar{B}_{\gamma} \defeq \sum_{t=0}^{\infty} \bar{A}^t \bar{D} l_t $, and leveraging Parseval's theorem to conclude the result $\bar{B}_{\gamma} = \frac{1}{2\pi} \int_{0}^{2\pi} (I-\e^{j\omega} \overline{A})^\inv \overline{D} L_\gamma(\e^{j\omega}) d\omega.$.
\QED
\subsection{Proof of Theorem \ref{thm:fixed-point}}
For $\gamma>\gamma_{\textrm{RO}}$, the KKT equation $N_\gamma(\ejw)= \frac{1}{4}\pr{I \+ \sqrt{I \+ 4 \gamma^\inv S_{\-,\gamma}^\ast(\ejw)  S_{\-,\gamma}(\ejw) }}^{\!2}$ is well defined. Thus, the mapping $F_4\circ F_3\circ F_2\circ F_1:\bar{B} \mapsto \bar{B}$ admits a fixed point $\bar{B}_\gamma$ for a fixed $\gamma$. 
The uniqueness of $\bar{B}_\gamma$ follows from the following:
The concavity of problem \eqref{eq:suboptimal_prob} in $\M_\gamma$ allows us to argue for uniqueness of $M_\gamma(\ejw)$. Given now that $M_\gamma(\ejw)=L_\gamma(\ejw)L_\gamma^\ast(\ejw)$, where $L_\gamma(\ejw)$ is a spectral factor of $M_\gamma(\ejw)$, which is causal and causally invertible, then $L_\gamma(\ejw)$ is unique up to a unitary transformation from the right. Fixing the choice of the unitary transformation in the spectral factorization (eg. positive-definite factors at infinity \citep{ephremidze_algorithmization_2018}) results in a unique $L_\gamma(\ejw)$. Given that $\bar{A}$ and $\bar{D}$ are fixed, $\bar{B}_\gamma=\frac{1}{2\pi} \int_{0}^{2\pi} (I-\e^{j\omega} \overline{A})^\inv \overline{D} L_\gamma(\e^{j\omega}) d\omega$ is unique.
\QED


\section{Spectral Factorization Method} \label{appdx:spect}
To perform the spectral factorization of an irrational function $N_\gamma(\ejw)$, we use the spectral factorization method presented in Algorithm \ref{alg:spect}. 
This algorithm returns samples of the spectral factor on the unit cirlcle. 

The method is efficient without requiring rational spectra, and the associated error term, featuring a purely imaginary logarithm, rapidly diminishes with an increased number of samples. It's worth noting that this method is designed specifically for scalar functions. 

\begin{algorithm}[ht!]
\caption{Spectral Factorization via DFT} 
\label{alg:spect}
\KwData{$N_\gamma(\mathrm{e}^{j\omega})>0$, and $N=2^k$}
\KwResult{$L_\gamma(\mathrm{e}^{j 2\pi n/N}) \quad \forall n\in [0,1,...,N-1]$}  
\textit{\textcolor{black}{/ / Compute the logarithm of $N_\gamma(\mathrm{e}^{j\omega})$}}\;

{\tiny [0]} $h(\mathrm{e}^{j\omega}) \gets \log(N_\gamma(\mathrm{e}^{j\omega}))$ \;  

\textit{\textcolor{black}{/ / Compute the inverse discrete Fourier transform (IDFT)}}\;

{\tiny [1]} $h_t \gets \text{IDFT}(h(\mathrm{e}^{j\omega}))  \hspace{0.35mm}\forall t\in[0,1,...,N/2]$\;

\textit{\textcolor{black}{/ / Compute the spectral factorization}}\;

{\tiny [2]} $L_\gamma(\mathrm{e}^{j 2\pi n /N}) \gets \exp\left(\frac{1}{2}h_0 + \sum_{t=1}^{N/2-1} h_t \exp\left(-j \frac{2\pi nt}{N}\right) + \frac{1}{2}(-1)^n h_{N/2}\right)$\;
\end{algorithm}

\section{Additional Simulations} \label{appdx:sim}
We compare the worst-case expected regret cost as defined in \eqref{def:worst_case_regret} of the DR-RO controller against the $H_2$, $H_\infty$ \cite{blackbook}, and RO \cite{sabag2021regret} controllers, considering the  worst-case disturbance of each controller for 3 other systems: [AC7], [NN3] and [RO6] (described in \cite{aircraft}), and we show the results in Table \ref{table:worstregret}.  
\begin{table}[!ht]
    \centering
    \small
    \setlength\tabcolsep{4pt} 
    \begin{tabular}{|c||c|c|c||c|c|c||c|c|c|} 
        \hline
        \textbf{ } & \multicolumn{3}{c||}{\textbf{[AC7]}} & \multicolumn{3}{c||}{\textbf{[NN3]}} & \multicolumn{3}{c|}{\textbf{[RO6]}} \\
        \hline
        \textbf{ } & \textbf{r=0.05} & \textbf{r=1} & \textbf{r=10} & \textbf{r=0.05} & \textbf{r=1} & \textbf{r=10} & \textbf{r=0.05} & \textbf{r=1} & \textbf{r=10} \\
        \hline \hline
        $\bm{H_2}$ & \textcolor{blue}{107.33}& 411.60& 14681& \textcolor{blue}{21.787}&	101.19&	4960.2&\textcolor{blue}{4.4425}	&21.070	&906.58  \\
        \hline
        \textbf{$\bm{H_\infty}$} &113.25&425.04& 14135& 55.346	&203.88&	6398.3&6.2662&	22.819&	706.44 \\
        \hline
        \textbf{RO} & 110.73&	401.76&\textcolor{green}{12153} & 36.348&	131.87&	\textcolor{green}{3990.9}&5.9874&	21.725&	\textcolor{green}{657.74} \\
        \hline
        \textbf{DR} & \textbf{\textcolor{blue}{107.30}} & \textbf{397.19} & \textbf{\textcolor{green}{12134}} & \textbf{\textcolor{blue}{21.781}} & \textbf{92.658} & \textbf{\textcolor{green}{3744.2}} & \textbf{\textcolor{blue}{4.4392}} & \textbf{19.337} & \textbf{\textcolor{green}{647.12}}\\
        \hline
    \end{tabular}
    \caption{The worst-case expected regret of each controller for different values of \(r\), for different systems: [AC7],[NN3] and [RO6] described in \cite{aircraft}. The cost of the DR-RO controller achieves the minimum in all systems, and it closely matches $H_2$ at small $r$ (see values in blue), while it closely matches RO at larger $r$ (see values in green).}
    \label{table:worstregret}
\end{table}
For all systems, for any $r$, the DR-RO controller achieves the minimum cost.
Moreover, its cost closely matches that of $H_2$ controller for smaller $r$ while it converges towards the behavior of the RO controller for larger $r$. This highlights the adaptability of the DR-RO controller which interpolates between $H_2$ and RO controllers, depending on the value of $r$ and achieves balanced performance across all input disturbances.
\end{document}